\documentclass[10pt,twocolumn,showpacs]{revtex4}
\usepackage{amsfonts}
\usepackage{amsmath}
\usepackage{graphicx}
\usepackage{dcolumn}
\usepackage{bm}
\usepackage{verbatim}
\usepackage{amsthm}
\usepackage{natbib}
\usepackage{doi}
\usepackage{hyperref}
\usepackage{bm}
\usepackage{epstopdf}

\newtheorem*{thm*}{Theorem}

\DeclareMathOperator{\Tr}{Tr}

\begin{document}

\title{Classical simulation of quantum dephasing and depolarizing noise}
\author{Daniel Crow}
\email{decrow@wisc.edu}
\author{Robert Joynt}
\email{rjjoynt@wisc.edu}
\affiliation{Physics Department, University of Wisconsin-Madison, Madison, WI, US\\
}
\date{\today}

\begin{abstract}
Dephasing decoherence induced by interaction of one qubit with a quantum
bath can be simulated classically by random unitary evolution without the
need for a bath and this random unitary evolution is equivalent to the
quantum case. \ For a general dephasing model and a single qubit system, we
explicitly construct the noise functional and completely specify the random
unitary evolution. \ To demonstrate the technique, we applied our results to
three paradigmatic models: spin-boson, central spin, and quantum impurity. \
For multiple qubits, we identify a class of generalized quantum dephasing
models that can be simulated classically. Finally we show that depolarizing
quantum models can be simulated classically for all dimensionalities of the
principal system.
\end{abstract}

\maketitle

\section{Introduction}

The notion of an open quantum system has received intense scrutiny in recent
years, motivated initially by questions about the location and nature of the
boundary (if any) between the quantum and classical worlds \cite%
{caldeiraleggett}.\ Further impetus has been given by the desire to minimize
decoherence in quantum information processing \cite{nielsenchuang, weiss}. \
Open quantum systems have been conceived in two quite distinct ways: a
system and bath considered as a single quantum entity with overall unitary
dynamics, or a quantum system acted on by random classical forces. \ The
first (system-bath or SB) approach is usually thought of as being more
general since it includes the transformation of information and the idea
that decoherence is connected to the entanglement of a system with its
environment \cite{ekert}. \ In the random classical (RC) force approach
these effects are absent. \ Furthermore, in the question of the
quantum-classical divide, it is often seen as important that the bath be
macroscopic since this precludes the complete measurement of the bath
degrees of freedom.

It is very clear from the Kraus operator representation that the SB approach
is more general than the RC one, since the operators need not be unitary in
SB. \ However, this does not illuminate what precisely is the role of
entanglement in separating the SB and RC methods. \ In fact, it has been
pointed out that SB models whose coupling is of dephasing type should always
be able to be rewritten as classical models \cite{strunz}. \ This is based
on general results concerning existence of certain positive maps \cite%
{Landau}. \ A specific example of a single qubit bath has been given \cite%
{strunz2}. \ Since SB entanglement is present in many dephasing models, this
implies that this entanglement is not the decisive feature of SB models that
distinguishes them from RC models.

In order to understand this particular aspect of the difference between
quantum theory and classical physics, we seek to find situations in which SB
models can be classically simulated. \ There are two aspects to this
program. \ The first is existence: which SB models can be classically
simulated? \ The second is construction: can we write down explicitly the RC
model that corresponds to a given SB model? \ The possibility of such
simulation of course does not imply that system-bath entanglement is not
important for quantum decoherence, but it does imply that the distinction
between quantum and classical noise is more subtle than has been previously
recognized.

The existence question for one qubit has been essentially settled in Refs.  
\cite{strunz} and  \cite{strunz2}. We extend this work by giving an explicit
construction method. \ As examples, we show explicitly how three
paradigmatic system-macroscopic dephasing bath models can be simulated by
subjecting the quantum 2-level system to a classical random force without
ever needing the bath. \ We do this by describing completely the classical
noise source that perfectly mimics the effect of entanglement of the qubit
with the bath (Sec. \ref{sec: models}). \ \ The three models are the popular spin-boson
model, a specific central spin model, and the quantum impurity model (Sec.
\ref{sec: single}). \ The last model is particularly interesting in this 
context since it
has a ``classical" and a ``quantum" phase 
\cite{lerner}. For multiple qubits, we can also give a construction method
for a certain class of multiple-qubit dephasing models, extending both
existence and construction results. \ \ The qubits can be interacting and
entangled, so the generalization is non-trivial (Sec. \ref{sec: multi}). \ Finally we show
that pure depolarizing models on qudits of arbitrary dimension can always be
simulated classically (Sec. \ref{sec: depol}).

\section{Classical Simulation of Quantum Dephasing Models}
\label{sec: models}

\subsection{Quantum Models}

The class of quantum models considered in this section consists of a single
qubit principal system subject to dephasing by an arbitrary quantum system.
\ Thus the total Hamiltonian of the composite system is 
\begin{equation}
H=H_{S}\left[ \sigma _{i}\right] +H_{B}\left[ \lambda _{i}\right] +H_{SB}%
\left[ \sigma _{i},\lambda _{i}\right]   \label{eq:sdef}
\end{equation}%
where $\sigma _{0,x,y,z}$ are the identity matrix $I$ and the Pauli matrices
that act on the qubit and the matrices $\lambda _{j}$ are some complete set
of Hermitian operators for the bath, taken for simplicity to be
finite-dimensional. \ $H_{B}\left[ \lambda _{i}\right] +H_{SB}\left[ \sigma
_{i},\lambda _{i}\right] $ are linear functions of the $\lambda _{i}.$ \ $%
\lambda _{0}$ is proportional to the identity, and $\Tr_{B}$ $\lambda _{i}=0$
for $i>0.$ \ For an $N$-level bath the $\lambda _{i}$ for $i>0$ could be
chosen as proportional to the $N^{2}-1$ generators of $SU\left( N\right) $
and $\lambda _{0}=I/N.$ \ $\Tr_{B}$ is the trace over bath variables. \ We
shall normalize the $\lambda _{i}$ by the condition that $\Tr_{B}\left[
\lambda _{i}\lambda _{j}\right] =\delta _{ij}$ for $i,j>0.$ \ $H_{B}\left[
\lambda _{i}\right] $ is the bath Hamiltonian. \ Infinite-dimensional baths
can be treated by the same method at the cost of some additional notational
complexity. \ 

We will focus on dephasing models, defined as those for which $\left[
H_{S},H_{SB}\right] =0.$

The total density matrix can be expanded as \ 
\begin{equation}
\rho =\sum_{ij}N_{ij}\sigma _{i}\otimes \lambda _{j}
\end{equation}%
and we can invert this equation : 
\begin{equation*}
N_{ij}=\frac{1}{2}\Tr\left[ \left( \sigma _{i}\otimes \lambda _{j}\right)
\rho \right] .
\end{equation*}%
The reduced density matrix of the system is%
\begin{equation*}
\rho ^{S}=\Tr_{B}\rho =\frac{1}{2}\sum_{i}n_{i}\sigma _{i}
\end{equation*}%
and $n_{i}=\left( 1,n_{x},n_{y},n_{z}\right) $ is the expanded Bloch vector
of the system. For quantum noise we have%
\begin{equation*}
\rho \left( t\right) =U\rho \left( 0\right) U^{\dag }
\end{equation*}%
with $U=\exp \left( -iHt\right)$ where we have set $\hbar = 1$. \ In terms
of components this is%
\begin{equation*}
N_{ij}\left( t\right) =\frac{1}{2}\sum_{kl}\Tr\left[ \left( \sigma
_{i}\otimes \lambda _{j}\right) U\left( \sigma _{k}\otimes \lambda
_{l}\right) U^{\dag }\right] N_{kl}\left( 0\right) .
\end{equation*}

We assume that the initial $\rho $ is in a product form $\rho \left(
0\right) =\rho ^{S}\left( 0\right) \otimes \rho ^{B}\left( 0\right) ,$ which
in components says that $N_{ij}\left( 0\right) =\frac{1}{2}n_{i}\left(
0\right) m_{j},$ where 
\begin{equation*}
m_{j}=\Tr_{B}\lambda _{i}\rho ^{B}\left( 0\right) .
\end{equation*}%
This is an important assumption, which will be used throughout the rest of
the paper. \ It is an open question whether the results below can be
generalized to non-product initial conditions. \ The dynamics of the system
can be rewritten as 
\begin{equation*}
N_{ij}\left( t\right) =\frac{1}{2}\sum_{kl}\Tr\left[ \left( \sigma
_{i}\otimes \lambda _{j}\right) U\left( \sigma _{k}\otimes \lambda
_{l}\left( 0\right) \right) U^{\dag }\right] n_{k}\left( 0\right) m_{l}
\end{equation*}%
and the Bloch vector of the qubit at time $t$ is given by 
\begin{align*}
n_{i}\left( t\right) & =2N_{i0}\left( t\right) \\
& =\sum_{kl}\Tr\left[ \left( \sigma _{i}\otimes \lambda _{0}\right) U\left(
\sigma _{k}\otimes \lambda _{l} \right) U^{\dag }\right] n_{k}\left(
0\right) m_{l}
\end{align*}%
Finally we have a linear relation between the initial and final expanded
Bloch vectors:%
\begin{equation*}
n_{i}\left( t\right) =\sum_{k}T_{ik}^{\left( Q\right) }n_{k}\left( 0\right)
\end{equation*}%
where the quantum transfer matrix is given by%
\begin{equation}
T_{ik}^{\left( Q\right) }\left( t\right) =\sum_{l}\Tr\left[ \left( \sigma
_{i}\otimes \lambda _{0}\right) U\left( \sigma _{k}\otimes \lambda
_{l}\right) U^{\dag }\right] m_{l}.  \label{eq:qtransfer}
\end{equation}%
One should note that this represents a linear relation between the initial
and final \emph{expanded} Bloch vectors. \ However, if construed as a
relation between the usual 3-dimensional Bloch vectors, it is an affine
relation, since if $n_{i}\left( t=0\right) =\left( 1,0,0,0\right) ,$ then
the initial state has zero ordinary Bloch vector, but the final state has $%
n_{i}\left( t\right) =T_{i0}^{\left( Q\right) },$ which is not zero. \
Expanding the Bloch vector to four components is just a convenient way of
including the affine part. \ 

The whole dynamics of the qubit is contained in the matrix $T_{ik}^{\left(
Q\right) }\left( t\right) .$ \ Given a Hamiltonian for the coupled system
and bath, $V$, the matrix $T_{ik}^{\left( Q\right) }\left( t\right) $ is
uniquely defined by Eq. \ref{eq:qtransfer}. \ More formally, let us define
the set $QM$ of quantum models specified by $H$\ and $U.$ \ Let $qm\in QM$
be a member of this set. \ Let $RS$ be the set of all evolutions of the
density matrix of the principal system and let $rs\in RS$ be a member of
this set. \ Then Eq. \ref{eq:qtransfer} defines a map $f:QM\rightarrow RS.$

\subsection{Classical Noise}

A\ 2-level system acted on by dephasing classical noise has the $2\times 2$
Hamiltonian%
\begin{equation}
H_{Cl}=-\frac{1}{2}B\sigma _{z}+\frac{1}{2}h\left( t\right) \sigma _{z}.
\label{eq:class}
\end{equation}%
$h\left( t\right) $ is a random classical (c-number) function of time that
represents an external source of noise. \ To completely specify the model we
need a probability functional $\mathcal{P}\left[ h\right] $ on the noise
histories. \ Then the system dynamics is given by%
\begin{equation}
\begin{aligned} \rho ^{\left( Cl\right) }\left( t\right) =& \int
\mathcal{D}\left[ h\left( t^{\prime }\right) \right] \mathcal{P}\left[
h\left( t^{\prime }\right) \right] \\ &\times U_{Cl}\left[ h\left( t^{\prime
}\right) \right] \rho ^{\left( Cl\right) }\left( 0\right) U_{Cl}^{\dag
}\left[ h\left( t^{\prime }\right) \right] . \label{eq:clev} \end{aligned}
\end{equation}%
This is a functional integral over all real functions $h\left( t^{\prime
}\right) $ defined on the interval $0\leq t^{\prime }\leq t.$ \ Each $%
h\left( t^{\prime }\right) $ is assigned a probability $\mathcal{P}\left[
h\left( t^{\prime }\right) \right] $ and 
\begin{equation*}
\int \mathcal{D}\left[ h\left( t^{\prime }\right) \right] \mathcal{P}\left[
h\left( t^{\prime }\right) \right] =1.
\end{equation*}%
\ Here the time-ordered exponential $U_{Cl}\left[ h\left( t\right) \right] =%
\mathcal{T}\exp \left( -i\int_{0}^{t}H_{Cl}\left( t^{\prime }\right)
~dt^{\prime }\right) $ is a $2\times 2$ matrix, unlike $U_{Q}\left( t\right)
,$ which is infinite-dimensional. In the dephasing case two Hamiltonians
taken at different times commute: $\left[ H_{Cl}\left( t\right)
,H_{Cl}\left( t^{\prime }\right) \right] =0$ and the time-ordering can be
dropped. \ Note the simplicity of the classical problem relative to the
quantum one. \ \ 

We now follow the steps of the discussion of the quantum models with the
appropriate modifications. Restating the results in terms of the $3\times 3$
transfer matrix $T^{\left( Cl\right) },$ we have, in exact analogy to the
section \ above,%
\begin{align*}
\rho ^{S}& =\frac{1}{2}\sum_{i=0}^{3}n_{i}\sigma _{i}, \\
n_{i}& =\frac{1}{2}\Tr\sigma _{i}\rho ^{S} \\
n_{i}\left( t\right) & =\sum_{j=0}^{3}T_{ij}^{\left( Cl\right) }\left(
t\right) n_{j}\left( 0\right)
\end{align*}%
where 
\begin{equation}
\begin{aligned} T_{ij}^{\left( Cl\right) }\left( t\right) =& \frac{1}{2}
\int \mathcal{D}\left[ h\left( t^{\prime }\right) \right] \mathcal{P}\left[
h\left( t^{\prime }\right) \right] \times \\ & \Tr\left\{
\sigma_{i}U_{Cl}[h\left( t^{\prime }\right) ]\sigma_{j} U_{Cl}^{\dag
}[h\left( t^{\prime }\right) ]\right\} \end{aligned}  \label{eq:ticl}
\end{equation}%
Note that $T_{i0}^{\left( Cl\right) }=\frac{1}{2}\Tr\sigma _{i}=\delta
_{i0}, $ so the classical model always gives a linear relation between
initial and final ordinary Bloch vectors--- there is no affine term. \ This
is another example of the fact that not all quantum models have a classical
analog. \ A\ simple example of an affine evolution is a qubit initially in
equilibrium at high temperature that is cooled by a bath. \ Then we have
that $\rho ^{S}\left( t=0\right) \approx I/2,$ whereas at long times the
system will be at thermal equilibrium at a lower temperature and the state
populations will be determined by the Boltzmann factor. \ This physical
process cannot be mimicked by classical external noise in the sense of
random unitary evolution. \ \ \ 

Define the set $CM$ of quantum models specified by $H_{cl}$\ and $\mathcal{P}%
.$ \ Then Eq. \ \ref{eq:ticl} defines a map $g:CM\rightarrow RS$. \ 

\subsection{Classical Simulation of Quantum Models}

To demonstrate the quantum-classical equivalence we need to prove the
existence of a functional $\mathcal{P}\left[ h\left( t\right) \right] $ such
that $\rho _{S}^{\left( Cl\right) }\left( t\right) =\rho _{S}^{\left(
Q\right) }\left( t\right) ,$ or, equivalently that $T_{ik}^{\left( Cl\right)
}\left( t\right) =T_{ik}^{\left( Q\right) }\left( t\right) .$ For dephasing
models we shall do more than this--- we shall give the explicit construction
of the classical model.\ 

In the construction we shall need only the discrete version of Eq. (\ref%
{eq:clev}), which is%
\begin{equation*}
\rho ^{\left( Cl\right) }\left( t\right) =\sum_{r}p_{r}U_{Cl}^{\left(
r\right) }\rho ^{\left( Cl\right) }\left( t=0\right) U_{Cl}^{\left( r\right)
\dag },
\end{equation*}%
where $r$ labels the possible noise histories and $U_{Cl}^{\left( r\right) }$
is the evolution operator for that history. The positive numbers $p_{r}$
satisfy \ $\sum_{r}p_{r}=1.$ \ The corresponding statement for the Bloch
vector is 
\begin{equation}
n_{i}\left( t\right) =\sum_{k,r}p_{r}O_{ik}^{\left( r\right) }\left(
t\right) n_{k}\left( 0\right) ,  \label{eq:ocl}
\end{equation}%
and the $O_{ik}^{\left( r\right) }$ are orthogonal matrices for every $r.$ \
Thus the problem of showing that a quantum noise model is actually classical
reduces to the problem of showing that the matrix $T_{ik}^{\left( Q\right) }$
can be written as a convex combination of orthogonal matrices. \ \ 

We shall consider the system-bath model for one qubit of Eq. (\ref{eq:sdef}%
). \ Without loss of generality we choose the total Hamiltonian%
\begin{equation}
H=-\frac{1}{2}B\sigma _{z}+H_{B}\left[ \lambda _{i}\right] +H_{SB}\left[
\lambda _{i}\right] \sigma _{z};  \label{eq:dephasing}
\end{equation}%
so that the qubit has no non-trivial dynamics and the noise is pure
dephasing.

In this model $\sigma _{z}$ is conserved, which implies that 
\begin{equation*}
T_{zz}^{\left( Q\right) }=1\; ;\;T_{xz}^{\left( Q\right) }=T_{zx}^{\left(
Q\right) }=T_{yz}^{\left( Q\right) }=T_{zy}^{\left( Q\right) }=0\text{.}
\end{equation*}%
We also have that $T_{00}^{\left( Q\right) }=1$ and $T_{i0}^{\left( Q\right)
}=T_{0i}^{\left( Q\right) }=0$ for $i>0,$ so we only need to calculate $%
T_{xx}^{\left( Q\right) },T_{yy}^{\left( Q\right) },T_{xy}^{\left( Q\right)
} $ and $T_{yx}^{\left( Q\right) }.$

First we note that the total evolution operator is 
\begin{equation}
U_{Q}=\mathcal{T}\exp \left( -i\int_{0}^{t}H\left( t^{\prime }\right)
dt^{\prime }\right) =u\left( t\right) +v\left( t\right) ~\sigma _{z},
\label{eq:timeordering}
\end{equation}%
where $u\left( t\right) $ and $v\left( t\right) $ are time-dependent bath
operators and $\mathcal{T}$ denotes time-ordering. \ Explicitly,%
\begin{equation}
\begin{aligned} u\left( t\right) &= \frac{1}{2}\Tr_{S}\left[ \mathcal{T}\exp
\left( -i\int_{0}^{t}H\left( t^{\prime }\right) dt^{\prime }\right) \right]
\label{eq:uandv} \\ v\left( t\right) &= \frac{1}{2}\Tr_{S}\left[ \sigma
_{z}\mathcal{T}\exp \left( -i\int_{0}^{t}H\left( t^{\prime }\right)
dt^{\prime }\right) \right] , \end{aligned}
\end{equation}%
where $\Tr_{S}$ is the trace over qubit variables. \ In Eq. \ref%
{eq:timeordering} the time ordering is essential, since the $\lambda _{i}$
do not commute with one another. \ Note that these expressions do not depend
in any way on having a finite-dimensional bath. \ Since $I=UU^{\dag }$ we
have%
\begin{align*}
uu^{\dag }+vv^{\dag }& =I \\
uv^{\dag }+vu^{\dag }& =0.
\end{align*}%
Expressing $T^{\left( Q\right) }$ in terms of $u$ and $v$ we find 
\begin{align}
T_{xx}^{\left( Q\right) }\left( t\right) & =T_{yy}^{\left( Q\right) }\left(
t\right)  \label{eq:tii} \\
& =\Tr\left[ \left( u^{\dag }+v^{\dag }\sigma _{z}\right) \left( \sigma
_{x}\otimes \lambda _{0}\right) \left( u+v\sigma _{z}\right) \left( \sigma
_{x}\otimes \rho ^{B}\left( 0\right) \right) \right]  \notag \\
& =2\Tr_{B}\left[ \left( u^{\dag }u-v^{\dag }v\right) \rho ^{B}\left(
0\right) \right]  \notag \\
T_{xy}^{\left( Q\right) }\left( t\right) & =-T_{yx}^{\left( Q\right) }\left(
t\right)  \label{eq:tij} \\
& =\Tr_{B}\left[ \left( u^{\dag }+v^{\dag }\sigma _{z}\right) \left( \sigma
_{x}\otimes \lambda _{0}\right) \left( u+v\sigma _{z}\right) \left( \sigma
_{y}\otimes \rho ^{B}\left( 0\right) \right) \right]  \notag \\
& =\Tr_{B}\left[ \left( -2iu^{\dag }v+2iv^{\dag }u\right) \rho ^{B}\left(
0\right) \right] ,  \notag
\end{align}%
and defining $c=T_{xx}^{\left( Q\right) }=T_{yy}^{\left( Q\right) }$ and $%
s=-T_{xy}^{\left( Q\right) }=T_{yx}^{\left( Q\right) }$, the matrix $%
T^{\left( Q\right) }$ has the form%
\begin{equation*}
T^{\left( Q\right) }=%
\begin{pmatrix}
1 & 0 & 0 & 0 \\ 
0 & c & -s & 0 \\ 
0 & s & c & 0 \\ 
0 & 0 & 0 & 1%
\end{pmatrix}%
.
\end{equation*}%
with positivity implying that $c^{2}+s^{2}=r^{2}\leq 1.$ \ The submatrix $%
T_{ij}^{\left( Q\right) }$ with $i,j=x,y$ is proportional to an orthogonal $%
2\times 2$ matrix. \ All effects of the bath on the system are summarized by
the quantities $c\left( t\right) $ and $s\left( t\right) .$\ 

The task is now to construct the equivalent classical model. \ In two
dimensions, an orthogonal matrix $M$ is characterized by a single unit
vector $\left( \cos \theta ,\sin \theta \right) ,$ so a convex sum of
orthogonal matrices is of the form%
\begin{equation*}
\sum\limits_{i}p_{i}M_{i}=%
\begin{pmatrix}
\sum\limits_{i}p_{i}\cos \theta _{i} & -\sum\limits_{i}p_{i}\sin \theta _{i}
\\ 
\sum\limits_{i}p_{i}\sin \theta _{i} & \sum\limits_{i}p_{i}\cos \theta _{i}%
\end{pmatrix}%
,
\end{equation*}%
with $p_{i}>0$ and $\sum_{i}p_{i}=1.$ \ This matrix is proportional to the $%
i,j=x,y$ submatrix of $T_{ij}^{\left( Q\right) }.$

We first write the vector $\left( c,s\right) $ as the convex sum of unit
vectors: 
\begin{equation*}
\begin{pmatrix}
c \\ 
s%
\end{pmatrix}%
=\frac{1}{2}%
\begin{pmatrix}
c+\beta s \\ 
s-\beta c%
\end{pmatrix}%
+\frac{1}{2}%
\begin{pmatrix}
c-\beta s \\ 
s+\beta c%
\end{pmatrix}%
\end{equation*}%
where%
\begin{equation}
r^{2}=c^{2}+s^{2}=\left[ T_{xx}^{\left( Q\right) }\right] ^{2}+\left[
T_{xy}^{\left( Q\right) }\right] ^{2}\;;\;\beta =\frac{\sqrt{1-r^{2}}}{r}
\label{eq: r and c}
\end{equation}%
This decomposition of the vector $\left( c,s\right) $ is not unique. \ This
implies that the mapping $g:CM\rightarrow RS$ is \textit{not }injective.

Then the classical evolution submatrix is written as 
\begin{equation*}
T^{\left( Cl\right) }=\frac{1}{2}%
\begin{pmatrix}
\cos \Phi _{1} & -\sin \Phi _{1} \\ 
\sin \Phi _{1} & \cos \Phi _{1}%
\end{pmatrix}%
+\frac{1}{2}%
\begin{pmatrix}
\cos \Phi _{2} & -\sin \Phi _{2} \\ 
\sin \Phi _{2} & \cos \Phi _{2}%
\end{pmatrix}%
.
\end{equation*}%
$T^{\left( Cl\right) }$ is the convex sum of rotations through the two
angles 
\begin{equation}
\begin{aligned} \Phi _{1}\left( t\right) =&\tan ^{-1}\left( \frac{s-\beta
c}{c+\beta s} \right) \\ &\text{and} \\ \Phi _{2}\left( t\right) =&\tan
^{-1}\left( \frac{s+\beta c}{c-\beta s} \right) . \end{aligned}
\label{eq:phi1 and phi2}
\end{equation}%
Hence, by comparison to Eq. (\ref{eq:ocl}), it defines a classical model. \
Define the fields $h_{1}\left( t\right) $ and $h_{2}\left( t\right) $ by 
\begin{equation}
\begin{aligned} h_{1} =&\frac{\partial \Phi _{1}}{\partial t}+B \\
&\text{and} \\ h_{2} =&\frac{\partial \Phi _{2}}{\partial t}+B. \end{aligned}
\label{eq:h1 and h2}
\end{equation}%
Then the equivalent classical model is given by%
\begin{equation*}
H_{Cl}=-\frac{1}{2}B\sigma _{z}+\frac{1}{2}h\left( t\right) \sigma _{z}
\end{equation*}%
and the probability distribution for $h$ is: $h\left( t\right) =h_{1}\left(
t\right) $ with probability $1/2$ and $h\left( t\right) =h_{2}\left(
t\right) $ with probability $1/2.$ \ \ More formally, $\mathcal{P}\left[
h\left( t\right) \right] =\left( 1/2\right) \delta \left[ h-h_{1}\right]
+\left( 1/2\right) \delta \left[ h-h_{2}\right] .$ \ Writing $\Phi
_{i}\left( t\right) =\int_{0}^{t}\left( -B+h_{i}\left( t^{\prime }\right)
\right) dt^{\prime }$ and following Eq. (\ref{eq:ticl}) we find%
\begin{align*}
T_{ij}^{\left( Cl\right) }=& \frac{1}{4}\Tr\left[ \sigma _{i}e^{\frac{i}{2}%
\sigma _{z}\Phi _{1}\left( t\right) }\sigma _{j}e^{-\frac{i}{2}\sigma
_{z}\Phi _{1}\left( t\right) }\right] \\
& +\frac{1}{4}\Tr\left[ \sigma _{i}e^{\frac{i}{2}\sigma _{z}\Phi _{2}\left(
t\right) }\sigma _{j}e^{-\frac{i}{2}\sigma _{z}\Phi _{2}\left( t\right) }%
\right]
\end{align*}

The contents of this section can be summarized by the following theorem and 
\emph{constructive} proof.

\begin{thm*}
The dynamics of the open quantum system given by Eq. (\ref{eq:dephasing})
can be simulated by the classical noise model given by Eq. (\ref{eq:class})
with $h(t)$ given by Eqs. (\ref{eq:h1 and h2}); the density matrix of the
qubit is the same for the two models at all times. \ 
\end{thm*}

\begin{proof} The quantum Hamiltonian $H$ of Eq. (\ref{eq:dephasing}) determines the
noise functions $u(t)$ and $v\left( t\right) $ in Eqs. (\ref{eq:uandv}) and
hence the matrix $T^{\left( Q\right) }$ in Eqs. (\ref{eq:tii}) and (\ref{eq:tij}) uniquely. \ The
classical fields $h_{1}\left( t\right) $ and $h_{2}\left( t\right) $ are
given explicitly in terms of the matrix elements of $T^{\left( Q\right) }$
by Eqs. (\ref{eq: r and c}) through (\ref{eq:h1 and h2}). \ 
$h_{1}\left( t\right) $ and $h_{2}\left( t\right) $
in turn define the classical Hamiltonian $H_{Cl}$ and noise probability functional  
$\mathcal{P}$ given in Eqs. (\ref{eq:class}) and (\ref{eq:clev}). \ 
The evolution of the qubit density matrix according to $H_{Q}$ with the usual 
partial trace is precisely the same as the qubit evolution according to the classical
Hamiltonian $H_{Cl}$ with the averaging over noise histories.
\end{proof}

\subsection{Equivalent Classical and Quantum Models}

Our terminology is to call a classical model equivalent to a quantum one if $%
rs\in RS,$ the evolution of the principal system, is the same. \ More
precisely, we say that $qm$ is equivalent to $cm$ if there exists $rs$ such
that $f\left( qm\right) =rs$ and $g\left( cm\right) =rs.$ \ We now complete
our characterization of $f$ using a standard method \cite{nielsenchuang},
that immediately implies that many quantum models correspond to any
particular $rs$. \ The proof is as follows. \ It is not restricted to single
qubits. \ Let the principal system have an $S$-dimensional Hilbert space $%
\mathcal{H}_{S}.$ \ An RC model has an EPS of the form%
\begin{equation}
\rho _{S}\left( t\right) =\sum_{\alpha }M_{\alpha }\left( t\right) \rho
_{S}\left( 0\right) M_{\alpha }^{\dag }\left( t\right) ,  \label{eq:1}
\end{equation}%
where $M_{\alpha }=\sqrt{p_{\alpha }}U_{\alpha },$ and $U_{\alpha }$ is the
unitary operator corresponding to the $\alpha $-th history of the classical
noise; $p_{\alpha }$ is the probability of that noise. \ $\rho _{S}$ is an $%
S\times S$ matrix. \ We then construct the equivalent quantum model as
follows. \ There is a bath, which is a $D$-dimensional quantum system with
Hilbert space $\mathcal{H}_{B}$ which has basis states $\left\{ \left\vert
0\right\rangle _{B},\left\vert 1\right\rangle _{B},\left\vert 2\right\rangle
_{B},...,\left\vert D-1\right\rangle _{B}\right\} .$ \ We associate a bath
state $\left\vert \alpha \right\rangle _{B}$ with each of the classical
histories $\alpha $. The bath is in the initial state $\left\vert
0\right\rangle _{B}$ and the total initial state is: 
\begin{equation}
\rho _{T}\left( t=0\right) =\rho _{S}\otimes \left\vert 0\right\rangle
_{B}~\left\langle 0\right\vert _{B}.  \label{eq:2}
\end{equation}%
The assumption that the environment is in a pure state is not too
restrictive since we can increase $D$ and purify the state if we wish. \ The
product form is restrictive, however, as already stated. \ We now need a
total evolution operator $U_{T}$ that is an $\left( SD\times SD\right) $
unitary matrix that acts on the coupled system and bath, i.e., it acts in
the Hilbert space $\mathcal{H}_{B}\otimes \mathcal{H}_{S}.$ \ We partially
define it by%
\begin{equation}
U_{T}\left( \left\vert n\right\rangle _{S}\otimes \left\vert 0\right\rangle
_{B}\right) =\sum_{\alpha =0}^{D-1}M_{\alpha }~\left\vert n\right\rangle
_{S}~\otimes \left\vert \alpha \right\rangle _{B}.  \label{eq:3}
\end{equation}%
This only gives the action of $U_{T}$ on states of the total system that are
of the form $\left( \left\vert n\right\rangle _{S}\otimes \left\vert
0\right\rangle _{B}\right) ,$ i.e., where the state of the system is
arbitrary but the environment is in the $\left\vert 0\right\rangle _{B}$
state. \ This is enough to show that $f$ is surjective. \ This definition
determines only $S$ columns, i.e., only $S^{2}D$ of the $S^{2}D^{2}$ of the
elements in $U_{T}.$ \ On states for which it is so far defined, $U$
preserves the inner product \cite{nielsenchuang}. \ To complete the
definition of $U,$ we can use Gram-Schmidt orthogonalization to get the
other $S\left( D-1\right) $ columns of $U_{T}.$ \ This process clearly
leaves considerable arbitrariness in the quantum model, so there are many
quantum models that correspond to a given classical model, i.e., $f$ is not
injective. \ We have already seen that $g$ is not injective, and earlier
work [6,7] had shown that $g$ is not surjective.

\section{Classical Simulation of Single-Qubit Dephasing Models}
\label{sec: single}

In this section we give three examples of the explicit construction of
random classical fields that correspond to well-known quantum models. \
Since the construction proceeds from the expression for the evolution of the
reduced density matrix, it is possible to do the construction exactly when
the quantum model is exactly solvable, (Sec. III A), and also when only
approximate solutions are known (Secs. III B and III C).

\subsection{Spin-Boson Model}

The spin-boson Hamiltonian is 
\begin{equation*}
H_{SB}=-\frac{1}{2}B\sigma _{z}+\sigma _{z}\sum_{k}\left( g_{k}b_{k}^{\dag
}+g_{k}^{\ast }b_{k}\right) +\sum_{k}\omega _{k}b_{k}^{\dag }b_{k}.
\end{equation*}%
The initial density matrix is $\rho \left( t=0\right) =\rho _{S}\left(
t=0\right) \otimes \rho _{B}\left( t=0\right) ,$ where $\rho _{S}\left(
t=0\right) $ is the initial state of the 2-level system and $\rho _{S}\left(
t=0\right) $ is the thermal state of the bath at temperature $1/\beta $. \
We are interested in the case of a macroscopic bath, and it is then
conventional to define the coupling function $J\left( \omega \right)
=4\sum_{k}\delta \left( \omega -\omega _{k}\right) \left\vert
g_{k}\right\vert ^{2}.$ \ Using $U_{Q}\left( t\right) =\exp \left(
-iH_{Q}t\right) ,$ the system dynamics is given by $\rho _{S}^{\left(
Q\right) }\left( t\right) =\Tr_{B}U_{Q}\left( t\right) \rho \left(
t=0\right) U_{Q}^{\dag }\left( t\right) $. \ The initially unentangled
system and bath are entangled by $U_{Q}\left( t\right) $. \ This model has
been well-studied and is exactly solvable \cite{breuer}, with the result
that the off-diagonal elements of $\rho _{S}\left( t\right) $ are
proportional to $e^{\Gamma \left( t\right) },$ where $\Gamma \left( t\right)
=-\int_{0}^{\infty }d\omega ~J\left( \omega \right) \coth \left( \beta
\omega /2\right) \left( 1-\cos \omega t\right) /\omega ^{2}<0.$ \ 

More explicitly, 
\begin{align*}
\rho^{(Q)} =& 
\begin{pmatrix}
\rho_{00} & \rho_{01}e^{\Gamma(t)+iBt} \\ 
\rho_{10}e^{\Gamma(t)-iBt} & \rho_{11}%
\end{pmatrix}
\\
=&\frac{\rho_{00}+\rho_{11}}{2}\sigma_{0} + \frac{\rho_{00}-\rho_{11}}{2}%
\sigma_{z} \\
& + \frac{1}{2}e^{\Gamma(t)}
\left(\rho_{01}e^{iBt}+\rho_{10}e^{-iBt}\right)\sigma_x \\
& + \frac{1}{2}e^{\Gamma(t)} \left(\rho_{10}e^{-iBt}-\rho_{01}e^{iBt}
\right)\sigma_y.
\end{align*}
$\rho_{ij}$ are the (time-independent) initial values for the elements of $%
\rho$. \ In terms of the Bloch vector we have 
\begin{align*}
n_{x}\left(t\right) &= e^{\Gamma(t)} \left(n_{x}\left(0\right) \cos Bt
+n_{y}\left(0\right)\sin Bt \right) \\
n_{y}\left(t\right) &= e^{\Gamma(t)} \left(-n_{x}\left(0\right) \sin Bt
+n_{y}\left(0\right)\cos Bt \right)
\end{align*}
We are free to consider the polarization vector along the $z$ direction by
change of basis and we can also begin with a pure state since total
evolution is simply proportional to the initial polarization vector. Written
in terms of elements of $T^{(Q)}$, 
\begin{align*}
c_{sb} &= e^{\Gamma(t)}\cos Bt \\
s_{sb} &= -e^{\Gamma(t)}\sin Bt \\
\beta_{sb} &= \frac{\sqrt{1-e^{2\Gamma(t)}}}{e^{\Gamma(t)}}.
\end{align*}
According to the prescription in the previous section, therefore, 
\begin{equation*}
H_{Cl}^{(sb)} = -\frac{1}{2}B\sigma_z+\frac{1}{2} h^{(sb)}\left(t\right)
\sigma_{z}
\end{equation*}
and the noise source $h^{(sb)}\left(t\right)$ is 
\begin{equation}
\begin{aligned} h_1^{(sb)}\left(t\right) &= \frac{\partial}{\partial
t}\tan^{-1}\left(\frac{s_{sb}-\beta_{sb}
c_{sb}}{c_{sb}+\beta_{sb}s_{sb}}\right) +B \;; \; p_1 = \frac{1}{2} \\
h_2^{(sb)}\left(t\right) &= \frac{\partial}{\partial
t}\tan^{-1}\left(\frac{s_{sb}+\beta_{sb}
c_{sb}}{c_{sb}-\beta_{sb}s_{sb}}\right) +B \; ; \; p_2 = \frac{1}{2}
\label{eq: spin boson} \end{aligned}
\end{equation}
where $p_{i}$ is the probability of $h_i$.

Eq. \ref{eq: spin boson} gives the result for a general coupling function $%
J\left( \omega\right) .$ \ A common choice for $J$ is the ohmic bath: $%
J(\omega)=A\omega e^{-\omega/\Omega}$, important in quantum optics. In this
case there is the exact result \cite{breuer} 
\begin{equation*}
\Gamma\left(t\right) \propto -\frac{1}{2}\ln \left(1+\Omega^{2}t^{2}\right)
- \ln \left( \frac{\sinh\left(t/\tau\right)}{t/\tau}\right).
\end{equation*}
where $\Omega$ is the cutoff frequency and $\tau = \frac{1}{\pi T}$ is the
thermal correlation time where we have taken $k_{B} = 1$. \ This expression
for $\Gamma\left(t\right)$ allows an exact calculation of $h_1\left(t\right)$
and $h_2\left(t\right)$ up to an overall scale factor. \ The results of this
calculation are plotted in Fig. \ref{fig:sb} for the case $\Omega \tau=20$
where we have taken the overall scale factor to be one. 
\begin{figure}[h]
\includegraphics[width=8cm]{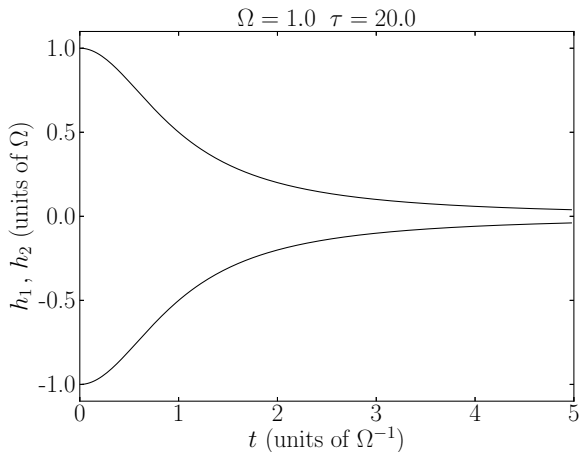}
\caption{Fields $h_1(t)$ and $h_2(t)$ for the spin-boson model.}
\label{fig:sb}
\end{figure}

The classical fields have a complicated form, with an initial quadratic
decay crossing over to exponential at longer times. \ The behavior for $%
t\gg\tau$ is due to thermal decay while the behavior for $t\ll\tau$ is
determined by the cutoff $\Omega$ and is due to fluctuations of the field 
\cite{breuer}.

\subsection{Central Spin Model}

The central spin Hamiltonian describes a qubit coupled to a bath of nuclear
spins. \ In the case of free induction decay, the Hamiltonian is 
\begin{equation*}
H_{CS}=-\frac{1}{2}B\sigma_z + H_{int} + H_{hf}
\end{equation*}
where 
\begin{equation*}
H_{int} = \sum_i\omega_i J_i^z + \sum_{i\neq
j}b_{ij}\left(J_i^{+}J_j^{-}-2J_i^zJ_j^z\right)
\end{equation*}
and 
\begin{equation*}
H_{hf} = \sigma_z \left(\frac{1}{2}\sum_i A_i J_i^z + \sum_{i\neq j}\frac{%
A_i A_j}{4B}J_i^{+}J_j^{-}\right).
\end{equation*}
$\omega_i$ are the nuclear Zeeman splittings, $b_{ij}$ contain the dipolar
interaction, and $A_i$ are the hyperfine couplings. \ We assume that the
initial state is a product state of the qubit and the equilibrium bath
state. \ For realistic conditions on $\omega_i$, $b_{ij}$, and $A_i$, the
qubit dynamics can be calculated approximately at experimentally relevant
time scales \cite{witzel,witzel2,witzel3}. \ The off-diagonal component of
the qubit density matrix is given by $\rho_{10}\left(t\right) =
\rho_{10}\left(0\right) D_{cs}\left(t\right)$ where 
\begin{equation*}
D_{cs}\left(t\right) = \frac{e^{i\arctan \left(\alpha t\right)-iBt}}{\sqrt{%
1+\alpha^2t^2}}
\end{equation*}
where $\alpha$ is a complicated function of the parameters of the model. \
The explicit dependence is found in \cite{witzel} and $\alpha$ determines
the relevant time scales since the model is valid only for $\alpha t \ll 1$.
\ In reference \cite{witzel}, $1/\alpha \approx 20\,\mu\text{s}$ for a GaAs
dot. \ 

The qubit gains an additional phase from the bath interaction. \ In terms of 
$T^{\left(Q\right)}$, we have 
\begin{align*}
c_{cs} &= \frac{1}{\sqrt{1+\alpha^2t^2}} \cos \left(\arctan\left(\alpha
t\right)-Bt\right) \\
s_{cs} &= \frac{1}{\sqrt{1+\alpha^2t^2}} \sin \left(\arctan\left(\alpha
t\right)-Bt\right) \\
\beta_{cs} &= \frac{\alpha t}{\sqrt{1+\alpha^2t^2}}.
\end{align*}
Then we can write 
\begin{align*}
\Phi_1 &= -Bt + \arctan\left(\alpha t\right) + \arcsin\left(\frac{\alpha t}{%
\sqrt{1+\alpha^2 t^2}}\right) \\
\Phi_2 &= -Bt + \arctan\left(\alpha t\right) - \arcsin\left(\frac{\alpha t}{%
\sqrt{1+\alpha^2 t^2}}\right)
\end{align*}
Noting that $\arctan \left(\alpha t\right) = \arcsin\left(\frac{\alpha t}{%
\sqrt{1+\alpha^2 t^2}}\right)$ and applying Eqs. (\ref{eq:h1 and h2}) we get 
\begin{align*}
h_1^{(cs)}\left(t\right) =& \frac{2\alpha}{1+\alpha^2t^2} \\
&\text{and} \\
h_2^{(cs)}\left(t\right) =& 0.
\end{align*}
This model shows a Lorentzian fall-off of one of the two possible fields but
the other one vanishes. \ Unlike the spin-boson model, these fields have a
non-zero time average which indicates that the central spin noise induces an
additional relative phase $\phi$ between the two system states. \
Furthermore, if we write the decoherence function as $r e^{i \phi}$, for $%
r,\phi \in \mathbb{R}$ so that $\rho_{10}(t) = \left( \rho_{10}(0)e^{-iBt}
\right)\left( r(t) e^{i\phi(t)} \right) $, then the central spin model gives
a simple relation between $r$ and $\phi$. \ Namely, $r=\cos \phi$ for all
values of $t$.

\subsection{Quantum Impurity Model}

The quantum impurity Hamiltonian is 
\begin{equation*}
H_{QI} = -\frac{1}{2}B\sigma_z + \frac{1}{2}v\left(d^{\dagger}d\right)%
\sigma_z + H_{B}
\end{equation*}
where 
\begin{equation*}
H_{B} = \epsilon_{0} d^{\dagger} d + \sum_k \left(t_k c_k^{\dagger} d + 
\text{H.c.}\right) + \sum_k\epsilon_k c_k^{\dagger} c_k.
\end{equation*}
$c_k^{\dagger}$ creates a reservoir electron and $d^{\dagger}$ creates an
impurity electron. \ $v$ represents the qubit-impurity coupling strength and 
$t_k$ are the tunneling amplitudes between the impurity and level $k$. \
Define the tunneling rate $\gamma = 2\pi\sum_k\left| t_k
\right|^2\delta(\epsilon_k-\epsilon_0)$. \ Assuming an initial product state
and an equilibrium bath, this model can be solved via numerically exact
techniques \cite{marquardt}. \ The off-diagonal components of the qubit
density matrix satisfy 
\begin{equation*}
\rho_{10}\left(t\right) = \rho_{10}\left(0\right) e^{-iBt} D_{qi}(t).
\end{equation*}
Using the numerical methods outlined in \cite{marquardt}, we compute $D(t) =
r(t) e^{i\phi(t)}$ for $r(t)>0$. \ Applying the method outlined earlier
yields 
\begin{equation*}
\Phi_i = -Bt+\phi(t)\pm\cos^{-1}(r(t))
\end{equation*}
and the classical noise source $h^{(qi)}\left(t\right)$ can be calculated
using Eqs. (\ref{eq:h1 and h2}). \ Numerical results are shown in Fig. \ref%
{fig:qi}. 
\begin{figure}[h]
\includegraphics[width=8cm]{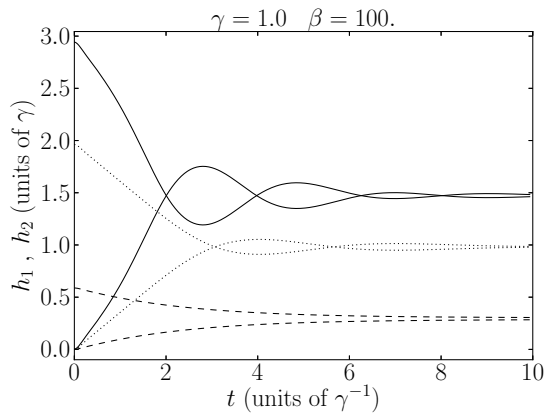}
\caption{Fields $h_1(t)$ and $h_2(t)$ for the quantum impurity model. The
solid line shows coupling $v=3$ (quantum phase). The dotted line shows $v=2$
(crossover). The dashed line shows $v=0.6$ (classical phase). \ Note the
oscillatory behavior of the stronger coupling, indicative of coherence
oscillation.}
\label{fig:qi}
\end{figure}

This model is of particular interest because it shows a ``classical'' phase
for $v\ll\gamma$ and a ``quantum'' phase for $v\gg\gamma$ with a crossover
region in between. \ The quantum phase is characterized by oscillation in
the coherence measures such as the visibility. \ These oscillations are
clearly caused by the oscillations in the noise fields as shown in Fig. \ref%
{fig:qi}. \ There is, however, nothing quantum about the fields at any value
of $v$; they are classical noise sources.

\section{Dephasing Of Multiqubit Systems}
\label{sec: multi}

In this section we generalize the one-qubit dephasing model to a class of
multiqubit models that satisfy a transitivity condition. \ In contrast to
the single-qubit case, there are multiqubit dephasing models that cannot be
classically simulated \cite{strunz2}. \ However, there are multiqubit
dephasing models that \emph{can} be simulated classically and that cannot be
reduced to independent qubits.

\subsection{Two-Qubit Dephasing}

We will begin with the two-qubit case. The total quantum Hamiltonian again
has the form 
\begin{equation*}
H=H_{S}+H_{B}+H_{SB}.
\end{equation*}%
Just as in the single qubit case, we define dephasing models by the
condition $\left[ H_{S},H_{SB}\right] =0.$ \ The system Hamiltonian $H_{S}$
can be expanded in the 15 generators of $SU\left( 4\right) .$ \ Now we note
that the subset of generators having the form $\sigma _{i}\otimes \sigma
_{i} $ is mutually commuting. \ Consider all models for which these
operators commute with the total Hamiltonian. \ That is, 
\begin{equation*}
\left[ H,\sigma _{i}\otimes \sigma _{i}\right] =0.
\end{equation*}%
for each $i$. \ Note that this class of models does \textit{not} reduce to
independent qubits. \ For example, the bath could be a set of phonons (or
other quantized fields) that modulates an isotropic Heisenberg coupling
between two spin 1/2 particles that have a time-independent Ising
interaction with each other. \ 

Once again we assume that the system and bath are initially in a product
state and that the full density matrix reduces to the system density matrix $%
\rho $. \ Note that since the generators $\sigma _{i}\otimes \sigma _{i}$
are mutually commuting they can be simultaneously diagonalized and the set
of eigenstates is equal to the set of maximally entangled Bell states. \ In
this basis the diagonal elements of $\rho $ are constant, i.e. $\rho
_{ii}(t)=\rho _{ii}(0)$. Thus, in this model, the populations of the Bell
states remain constant in time. Note that these qubits are entangled and
thus are not independent. We will continue to work in this basis for the
remainder of the section.

For a classical noise source to simulate the system-bath Hamiltonian and
preserve the diagonal elements, each noise source must be diagonal. \ Note
that in the usual basis, this is equivalent to the condition that each noise
term is a linear combination of the generators $\sigma _{i}\otimes \sigma
_{i}$. \ If this is the case, then the time evolution of the reduced density
matrix is given by $\rho _{ij}(t)=r_{ij}(t)\rho _{ij}(0),$ (no summations)
and $r_{ij}$ satisfies the conditions $r_{ij}(t)=r_{ji}(t)^{\ast }$, $%
\left\vert r_{ij}(t)\right\vert \leq 1$, and $r_{ii}(t)=r_{ij}(0)=1$.

To simulate this quantum system, we need classical noise sources $H_{\alpha }
$ acting with a probability distribution $p(\alpha )$ for $\alpha \in 
\mathbb{R}$. \ The corresponding evolution of the two-qubit density matrix
is given by 
\begin{equation*}
\rho (t)=\int_{-\infty }^{\infty }p(\alpha )U_{\alpha }(t)\rho (0)U_{\alpha
}^{\dagger }(t)d\alpha .
\end{equation*}%
In our basis of Bell states, the set of valid noise terms $H_{\alpha }$ is
simply the set of diagonal matrices. Thus, define the diagonal entries of
each $H_{\alpha }$ as $\alpha d_{i}(t)$ so that the set $\left\{ H_{\alpha
}\right\} $ is characterized by a single parameter, $\alpha $. \ This
assumption is a particularly simple case.

Time evolution is given by diagonal unitary operators $U_{\alpha }$ with
diagonal entries $\exp \left( -i\alpha \theta _{i}(t)\right) $ for 
\begin{equation*}
\theta _{i}(t)=\int_{0}^{t}d_{i}(t^{\prime })dt^{\prime }.
\end{equation*}%
Then the action $U_{\alpha }\rho U_{\alpha }^{\dagger }$ gives 
\begin{equation*}
\rho _{ij}\rightarrow e^{-i\alpha \gamma _{ij}(t)}\rho _{ij}
\end{equation*}%
where $\gamma _{ij}(t)=\theta _{i}(t)-\theta _{j}(t)$. \ Applying this
result and averaging the result over the probability distribution gives 
\begin{align}
\rho _{ij}(t)=& \int_{-\infty }^{\infty }p(\alpha )\rho _{ij}(0)e^{-i\alpha
\gamma _{ij}(t)}d\alpha   \notag \\
=& \rho _{ij}(0)\tilde{p}\left( \gamma _{ij}(t)\right) 
\end{align}%
where $\tilde{p}(t)=\int_{-\infty }^{\infty }p(\alpha )e^{-i\alpha t}d\alpha 
$. \ That is, for the model specified above, 
\begin{equation}
r_{ij}(t)=\tilde{p}(\gamma _{ij}(t)).
\end{equation}%
$r_{ij}$ is just the Fourier transform of the probability distribution $%
p\left( \alpha \right) $, and we can choose $p\left( \alpha \right) $ as we
wish. \ This completes the classical description of the process. \ 

The class of quantum systems susceptible to this construction is not very
broad: the fact that $\gamma _{ij}$ is a difference of two functions is
quite restrictive. \ $\gamma _{ij}\left( t\right) $ must satisfy the
transitivity condition:~$\gamma _{ij}\left( t\right) =\gamma _{ik}\left(
t\right) +\gamma _{kj}\left( t\right) $ for all $i,j,k,$ which leads in turn
to a restrictive set of conditions on the $r_{ij}(t)$. \ Such models are
simple enough to construct however, since we only need to specify the $%
\theta _{i}(t)$ and transitivity is automatically satisfied. However, the
transitivity condition limits us to three noise functions $\theta _{i}$ out
of a possible six $\gamma _{ij}$.

\subsection{Generalization to Many Qubits}

The results for two qubits can be extended to more qubits in a
straightforward manner. Let the principal system consist of $n$ qubits. \ In
the quantum case, we consider the principal system coupled to a bath with
evolution governed by a system-bath Hamiltonian $H_{SB}$.

In a system of $n$ qubits, the Pauli operators \cite{nielsenchuang} can be
formed as the tensor products of 2-qubit Pauli matrices. \ Including the
identity, there are $4^{n}$ Pauli operators that act on the space of $n$
qubits. \ It has been shown \cite{lawrence} that the $4^{n}-1$ Pauli
operators (excluding the $2^{n}\times 2^{n}$ identity) can be partitioned
into $2^{n}+1$ unique sets of $2^{n}-1$ pairwise commuting operators. \ Let $%
A$ be a subset of the Pauli operators consisting of $2^{n}-1$ commuting
operators. \ For example, $A$ could be the subset defined by all tensor
products of $\sigma _{0}$ and $\sigma _{3}$ (excluding the identity).

The Pauli operators $O_{i}\in A$ are mutually commuting and can be
simultaneously diagonalized. \ Along with the identity, we have $2^{n}$
independent operators. \ Working in a basis such that each $O_{i}$ is
diagonal, these operators span the diagonal matrices. \ Consider a
system-bath Hamiltonian $H_{SB}$ such that 
\begin{equation*}
\left[ H_{SB},O_{i}\right] =0
\end{equation*}%
for each $O_{i}$. Since the $O_{i}$ are simultaneously diagonalized, there
is no population transfer in the basis of simultaneous eigenstates and it is
clear that the above conditions produce dephasing noise and we can define
the analogous quantities $r_{ij}(t)$ characterizing the evolution of the
principal system. However, as in the two-qubit case, not all of the quantum
models satisfying these conditions can be simulated classically.

We can consider a random noise model acting on the principal system
specified by $\left\{ H_{\alpha }\right\} $ and $p(\alpha )$ with all of the
noise terms $H_{\alpha }$ equal to linear combinations of $O_{i}\in A$.
Proceeding as in the two-qubit case, define the analogous quantities $\gamma
_{ij}(t)$. These $\gamma _{ij}$ remain subject to the transitivity condition 
$\gamma _{ij}(t)=\gamma _{ik}(t)+\gamma _{kj}(t)$. As in the 2-qubit case, a
sufficient condition for classical simulation is that $r_{ij}(t)=\tilde{p}%
(\gamma _{ij}(t))$. For many qubits, the transitivity condition on $\gamma
_{ij}$ becomes increasingly restrictive. \ Once again, however, we only need
to specify the $\theta _{i}(t)$ in order to construct such a model. With $n$
qubits, this amounts to specifying $2^{n}-1$ degrees of freedom, out of a
possible $\left( 2^{n-1}\right) \left( 2^{n}-1\right) $.

\section{N-dimensional depolarization}
\label{sec: depol}

Given a principal quantum system in $N$ dimensions with initial density
matrix $\rho _{0}$, the depolarization channel is given by 
\begin{equation*}
\rho (t)=(1-p(t))\rho _{0}+p(t)\left( \frac{1}{N}\right) I
\end{equation*}%
where $p(t)\in  [0,1]$ for all $t$ and $p(0)=0$. \ $I$ is the $N\times
N$ identity matrix. \ In this section we drop the subscripts on $\rho .$ \ 

A quantum model for depolarization of $n$ $(N=2^{n})$ qubits is given by the
Kraus operators $M_{\underline{i}}$ where $\underline{i}=\left(
i_{1},i_{2},\ldots ,i_{n}\right) $ with $i_{j}\in \left\{ 0,1,2,3\right\} $ 
\begin{equation}
M_{\underline{0}}(t)=\sqrt{1-\frac{4^{n}-1}{4^{n}}p(t)}I
\end{equation}%
and for $\underline{i}\not=\underline{0}$ 
\begin{equation}
M_{\underline{i}}(t)=\frac{\sqrt{p(t)}}{2^{n}}\left( \sigma _{i_{1}}\otimes
\sigma _{i_{2}}\otimes \cdots \otimes \sigma _{i_{n}}\right) .
\end{equation}%
Then 
\begin{equation}
\rho (t)=\sum_{\underline{i}}M_{\underline{i}}\rho _{0}M_{\underline{i}%
}^{\dagger }
\end{equation}%
holds for $\rho (t)$ defined above \cite{nielsenchuang}.

We now present a classical model for depolarization in arbitrary dimension.
\ Consider the unitary group $SU(N)$. $\ N\ $need not be a power of 2, so the
following construction holds not only for qubit arrays, but also for qudits.
Each time-independent $U\in SU(N)$ can be expanded in terms of its
eigenvectors $\left\vert j\right\rangle $ as 
\begin{equation*}
U=\sum_{j}e^{-id_{j}}\left\vert j\right\rangle \left\langle j\right\vert 
\end{equation*}%
with $d_{j}\in \left( -\pi ,\pi \right] $. \ For a fixed $U$, define $H_{U}$
as 
\begin{equation*}
H_{U}=\sum_{j}d_{j}\left\vert j\right\rangle \left\langle j\right\vert 
\end{equation*}%
so that 
\begin{equation*}
U=\left. e^{-iH_{U}t}\right\vert _{t=1}.
\end{equation*}%
Note that each $H_{U}$ can be expanded as 
\begin{equation*}
H_{U}=\sum_{r}h_{r}\lambda _{r}
\end{equation*}%
where $h_{r}$ are real and $\lambda _{r}$ are the generators of $SU(N)$.
Then consider the Hamiltonian defined by the set $\left\{ H_{U}\right\} $
with the probability distribution equal to the uniform distribution in Haar
measure over $SU(N)$. \ For convenience, choose 
\begin{equation*}
\int_{SU(N)}dU=1
\end{equation*}%
where the integral is taken with respect to the Haar measure. \ Time
evolution is then given by 
\begin{equation}
\rho (t)=\int_{SU(N)}e^{-iH_{U}t}\rho (0)e^{iH_{U}t}dU.
\label{depolarization}
\end{equation}

We need to show that this time evolution is purely depolarizing. Consider an arbitrary unitary transformation $V$ such that 
\begin{equation*}
V\rho _{0}V^{\dagger }=\rho _{0}.
\end{equation*}%
Then we can write that
\begin{equation*}
\rho (t)=\int_{SU(N)}e^{-iH_{U}t}V\rho _{0}V^{\dagger }e^{iH_{U}t}dU
\end{equation*}%
from which we can easily obtain 
\begin{equation*}
V^{\dagger }\rho (t)V=\int_{SU(N)}V^{\dagger }e^{-iH_{U}t}V\rho
_{0}V^{\dagger }e^{iH_{U}t}VdU.
\end{equation*}%
Since the above integral runs over the entire unitary group $SU(N)$,
conjugation by $V$ simply amounts to a reparameterization of the integral. \
Thus, we can write 
\begin{equation}
V^{\dagger }\rho (t)V=\int_{SU(N)}e^{-iH_{U}t}\rho _{0}e^{iH_{U}t}dU=\rho
(t).
\end{equation}%
Hence, any rotation leaving the initial polarization vector fixed must also
leave $\rho (t)$ fixed, and the polarization vector of $\rho (t)$ must be
parallel to that of $\rho _{0}$.

Now, consider the case $t=1$. \ From our initial definition of $H_{U}$, we
see that 
\begin{equation*}
\rho (t=1)=\int_{SU(N)}U\rho _{0}U^{\dagger }dU.
\end{equation*}%
Further consider an arbitrary unitary operator $W$. Then conjugation gives 
\begin{align*}
W\rho (t& =1)W^{\dagger }=\int_{SU(N)}WU\rho _{0}U^{\dagger }W^{\dagger }dU
\\
& =\int_{SU(N)}U\rho _{0}U^{\dagger }dU
\end{align*}%
since if $U$ runs over the entire unitary group than so does $WU$. \ Thus,
we see that 
\begin{equation*}
W\rho (t=1)W^{\dagger }=\rho (t=1).
\end{equation*}%
Since this holds for arbitrary unitary operator $W$, we must have that $\rho
(1)$ is proportional to the identity. \ Thus, $\rho (1)$ is the totally
mixed state and we see that the model presented is indeed a depolarizing
channel. \ 

For the single qubit ($N=1$) case, the integral can be evaluated explicitly.
\ Without loss of generality, assume that the initial qubit state is
specified by 
\begin{equation*}
\rho (t=0)=\frac{1}{2}\left( I+\sigma _{z}\right) .
\end{equation*}%
Then 
\begin{equation*}
\rho (t)=\frac{1}{2}\left( I+n_{z}(t)\sigma _{z}\right) .
\end{equation*}%
The integration in Eq. \ref{depolarization} reduces to an integral over the $%
3$-sphere $S^{3}$. \ It can be evaluated explicitly and one finds 
\begin{equation}
n_{z}(t)=\frac{1}{3}+\frac{\sin (2\pi t)}{3\pi (t-t^{3})}.
\end{equation}%
This expression indeed satisfies $n_{z}(t=0)=1$ and $n_{z}(t=1)=0$. \
Somewhat surprisingly, this function has a root near $t=0.77$ and has a
limiting value of $1/3$ as $t\rightarrow \infty $. \ Because of the
unexpected root, proper depolarization behavior ends before $t=1$, however a
simple reparameterization of time and the corresponding redefinition of
noise terms $H_{\alpha }$ allows for arbitrary depolarizing behavior. \ The
graph of $n_{z}(t)$ is shown in Fig. \ref{depo}. 
\begin{figure}[h]
\includegraphics[width=8cm]{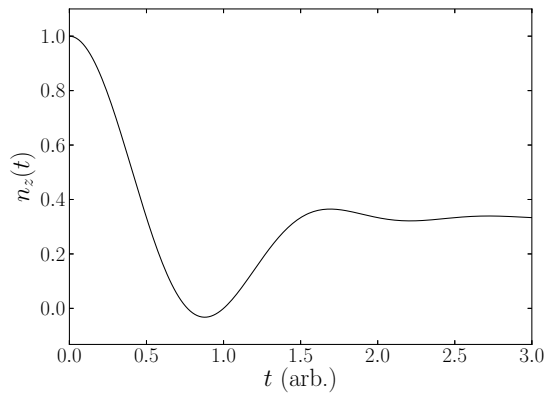}
\caption{The $z$ component of the Bloch vector as a function of time for
single qubit classical depolarization. \ Note that proper depolarization
ends at the first root ($t\approx 0.77$). }
\label{depo}
\end{figure}

In this classical model the noise histories belong to an uncountably
infinite set. \ It is also possible to construct a classical depolarization
model using a finite set of histories. \ The Clifford group $C\left(
N\right) $ is a subgroup of $SU\left( N\right) $ defined as the stabilizer
group of the Pauli group; the latter consists of tensor products of the
Pauli matrix multiplied by the units of the complex numbers. \ It has been
shown that $C\left( N\right) $ is a unitary 2-design, meaning that averages
over the Clifford group of polynomial expressions of degree less than or
equal to two (such as $\rho _{0})$ are equal to uniform averages over the
Haar measure \cite{clifford, eisert, seymour}. \ This means that the above
derivation of the depolarizing evolution may be repeated with the
substitution 
\begin{equation*}
\int_{SU(N)}f\left( U\right) dU\rightarrow \frac{1}{\left\vert C\left(
N\right) \right\vert }\sum_{C}f\left( C\right) .
\end{equation*}%
Since the members of $C\left( N\right) $ are unitary, evolution of $\rho $
according to 
\begin{equation*}
\frac{1}{\left\vert C\left( N\right) \right\vert }\sum_{U\in C}C\rho C^{\dag
}
\end{equation*}%
is a RC evolution. \ The construction of the equivalent depolarizing quantum
system-bath model then proceeds as in Eqs. \ref{eq:1}-\ref{eq:3}.

\section{Conclusion}
\label{sec: conclusion}

We have shown how to construct classical simulations of decoherence arising
from the interactions of a single qubit and an external bath for the pure
dephasing case. \ We demonstrated that the noise functional takes on the
particularly simple form of two uniformly distributed random unitary
evolutions and gave the explicit expression for the corresponding
Hamiltonian fields. \ We offered exact results of the calculation of these
fields for the spin-boson model, approximate results for the central spin
model, and numerically exact results for the quantum impurity model. \ 

Classical simulation of quantum noise acting on multiple qubits and qudits
is also possible in certain cases. For two interacting qubits, we showed
that the simulation is possible if the interaction is sufficiently simple
and we gave the corresponding classical noise model. \ We demonstrated that
the depolarization channel can be simulated classically for a state space of
arbitrary dimension, using well-known properties of the Clifford group. \ 

Philosophically, it is surprising that classical simulation is mainly
possible when the decoherence arises from phase randomness of quantum states
and is more difficult when the decoherence comes from randomness in the
population of those states, since phase is a characteristic of quantum
mechanics that is not shared with classical mechanics, while the concept of
population is common to both. \ 

\begin{acknowledgments}
We acknowledge with pleasure illuminating discussions with W. Strunz and
B.L. Hu. \ This project was funded by DARPA-QuEst Grant No. MSN118850.
\end{acknowledgments}

\bibliographystyle{apsrev4-1}
\bibliography{bibliography}

\end{document}